\documentclass[runningheads]{llncs}
\usepackage{graphicx}
\usepackage{color}
\usepackage{xcolor}
\usepackage{amsfonts}
\usepackage{amssymb}
\usepackage[normalem]{ulem}
\usepackage{svg}
\usepackage{enumitem}
\usepackage{nicematrix}
\usepackage{booktabs}
\usepackage{caption}
\usepackage[nocompress]{cite}

\usepackage{hyperref}

\newcommand{\coloneqq}{\mathrel{:=}}
\definecolor{orcidgreen}{RGB}{166,206,57}
\newcommand{\orcidicon}{\includegraphics[width=0.20cm]{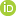}}
\def\orcidID#1{\renewcommand{\thefootnote}{\roman{footnote}}
  \unskip$^{\orcidicon}$\footnote{\orcidicon\color{orcidgreen}\,\scriptsize #1\color{black}\vspace*{-2pt}}\unskip
  \renewcommand*{\thefootnote}{\arabic{footnote}}\unskip
}

\begin{document}

\author{%
  \setcounter{footnote}{-1}
  Amos~Onn\inst{1,2}$^{,*}$\thanks{$^*$Corresponding author}
  \fnmsep\orcidID{0000-0002-8403-0754}
  \and
  Tzipy~Marx\inst{3}
  \and
  Liming~Tao\inst{4}
  \fnmsep\orcidID{0000-0002-3034-925X}
  \and
  Tamir~Biezuner\inst{3}
  \fnmsep\orcidID{0000-0002-3642-0977}
  \and
  Ehud~Shapiro\inst{3}
  \and
  Christoph~A.~Klein\inst{1,5}
  \fnmsep\orcidID{0000-0001-7128-1725}
  \and
  Peter~F.~Stadler\inst{2,6-10}
  \fnmsep\orcidID{0000-0002-5016-5191}
  }

\institute{
  Chair of Experimental Medicine and Therapy Research, University of
  Regensburg, Universit{\"a}tsstra{\ss}e 31, D-93053 Regensburg, Germany,
  \email{amos.onn@klinik.uni-regensburg.de},
  \email{christoph.klein@klinik.uni-regensburg.de}
  \and 
  Bioinformatics Group, Faculty of Mathematics and Computer Science, and
  Interdisciplinary Center for Bioinformatics, University of Leipzig,
  H{\"a}rtelstrasse 16-18, D-04107 Leipzig, Germany,
  \email{amos.onn@uni-leipzig.de},
  \email{studla@bioinf.uni-leipzig.de}
  \and
  Department of Computer Science and Applied Mathematics,
  Weizmann Institute of Science, 234 Herzl Street, POB 26, Rehovot 7610001,
  Israel,
  \email{tzipy.marx@gmail.com},
  \email{tamir.biezuner@weizmann.ac.il},
  \email{ehud.shapiro@weizmann.ac.il}
  \and
  Cellular Tissue Genomics, Genentech, 94080 South San Francisco, CA, USA,
  \email{taoliming.too@gmail.com}
  \and
  Fraunhofer Institute for Toxicology and Experimental Medicine Regensburg,
  Am BioPark 13, D-93053 Regensburg, Germany
  \and 
  Max Planck Institute for Mathematics in the Sciences, Inselstra{\ss}e 22
  D-04103, Leipzig, Germany
  \and
  Institute for Theoretical Chemistry, University of Vienna,
  W{\"a}hringerstra{\ss}e 17, A-1090 Vienna, Austria
  \and
  Facultad de Ciencias, Universidad Nacional de Colombia,
  Bogot{\'a}, Colombia
  \and
  Center for non-coding RNA in Technology and Health, University of Copenhagen,
  Ridebanevej 9, DK-1870 Copenhagen, Denmark
  \and
  Santa Fe Institute, 1399 Hyde Park Rd., 87501 Santa Fe, New Mexico, USA
}

\title{Mutational Dynamics of Very Short Tandem Repeats
}

\authorrunning{A.\ Onn \emph{et al.}}
\titlerunning{STR Mutation Model}

\maketitle              
\begin{abstract}
  Short tandem repeats (STRs) are low-entropy regions in the genome,
  consisting of a short (1-6 bp) unit that is consecutively repeated
  multiple times. They are known for high mutational instability, due to
  so-called stutter-mutations, in which the number of units in the run
  increases or descreases. In particular, STRs with repeat unit length of
  1-2 bp are prone to mutate even within several cell divisions. The
  extremely rapid accumulation of variation makes them interesting
  phylogenetic markers for retrospective single-cell lineage
  reconstruction. Here we model their mutational dynamics at the level of
  individual repeat unit motif and then aggregate length variations over
  many STR loci with the aim of obtaining a very fast ``molecular clock''.
  We calibrate our model based on several datasets with known lineage
  structure prepared from cultured cells. We also verify the error margins
  around our estimated parameters for the model using simulated data. We
  find that the mutational dynamics of STRs are reasonably consistent for a
  given cell-line, but vary among different ones. This suggests that the
  dynamics are not entirely explained by mutations in caretaker genes,
  rather, various other factors play a role --- possibly tissue origin and
  differentiation state. Further data and research is necessary to assess
  their relative effects.
  
  \keywords{short tandem repeats, microsatellites, STR, MS,
    lineage, single cell}
\end{abstract} 

\section{Introduction}

Reconstructing the lineage of single cells is a problem of ongoing
interest, with implications in various fields~\cite{Wagner2020Lineage,
  Tang2020Integrating,ChapalIlani2013Comparing}. At the end of the
spectrum, reconstructing the entire lineage tree of an organism can shed
light on processes of cell differentiation and organisation. This has been
done for the small organism \emph{Caenorhabditis elegans} (631--633 cells),
via direct manipulation of the cells and observation under a
microscope~\cite{celegans_1983}. This method is impractical for larger
organisms, as is obtaining a fully resolved lineage tree. However even
reconstructing subtrees of hundreds of cells could yield results of
interest, informing about the development of specific parts, tissues or
diseases --- in particular, cancer development. In recent years, genetic
recording methods have been used for reconstructing such partial cell
lineages in model organisms~\cite{crispr_2016,Kretzschmar_2012,Lu_2011,
  Frieda2017Synthetic,Kalhor2017Rapidly,Raj2018Simultaneous,
  Schmidt2017Quantitative,McKenna2016WholeOrganism}. Such methods remain
nonetheless irrelevant for human samples, where such a manipulation is not
tenable. In this case, retrospective reconstruction of single cell lineage
has to rely on \emph{a posteriori} observation of naturally-occuring
somatic mutations.

One possibility would be to use single nucleotide polymorphisms (SNPs). The
resolution of such data is however limited by the comparably low mutations
rates, around $10^{-8}$ per site per cell division~\cite{Wang_2012},
implying that even genome-wide sequencing with high coverage --- and the
associated high effort and cost --- will capture only a very small number
of somatic mutations~\cite{Hou_2012,Xu_2012}. A possibly cost-saving
alternative is to focus on ``hot-spot'' SNPs; however, this results in a
panel of SNPs partially driven by selective pressure, thus invalidating the
assumption of independence of mutations and skewing the
results~\cite{Johnson_2011}. This is particularly relevant in the case of
lineage-tracing of cancer, which is one of the major
use-cases~\cite{Hess_2019}. Other possible mutations, such as L1 retro
transposition events, copy number variants, and mitochondrial
heteroplasmy~\cite{Lareau2023Mitochondrial} might be subject to similar
difficulties.

To overcome some of these problems, we have developed an assay for
reconstructing lineage based on sequencing particular regions of the genome
known as short tandem repeats (STRs, also known as
microsatellites)~\cite{cost_effective_lineage_2016,stutter_2019,om_lineage_2021}.
These are low-entropy regions, where a short (1--6 bp) motif of nucleotides
repeats itself in direct consecution, with numerous ($\ge 5$) copies of the
repeat unit following each other. Due to the low entropy, during DNA
duplication, the strands may detach and reattach at an offset of one or
more repeat units, leading to a so-called stutter mutation. As a
consequence the number of repeat units, also referred to as the length of
the STR, increases or decreases. The rate for these mutations is several
orders of magnitude higher than that of random single point mutations:
estimates range from $10^{-3}$--$10^{-5}$ per site per cell division for
di-repeats~\cite{Willems_2016}. In addition, STRs are considered
evolutionary neutral~\cite{Ellegren_2004}. The extremely rapid accumulation
of differences in repeat length makes STRs, in particular those with very
short units, excellent candidates for tracing single cell lineage. The high
mutability, however, also creates difficulties for the analysis of such
data:
\begin{itemize}
\item[(1)] The STR lengths observable in sequencing data differ from the
  \emph{in vivo} state because the unavoidable amplification steps during
  library preparation introduce additional stutter mutations.
\item[(2)] Mutation rates differ substantially among different loci and
  cell types.
\end{itemize}

The first point has been addressed in a previous work~\cite{stutter_2019};
our goal is to devise a continuous-time Markov model describing \emph{in
vivo} temporal evolution of STR length, focusing on the second point. This
model may then be used as a fast molecular clock, enabling estimations of
distance in cell-divisions between single-cell samples. This distance could
then further be used for reconstructing lineage trees of these cells; the
Markov model could also be used directly to assess and maximize the
likelihood of the tree topogies. Such trees, in addition to a topological
structure, also have reliable edge lengths, and thus contribute toward a
better understanding of the evolution of the analysed tissues.

Our aim is to allow for different rates of evolution at each locus. The
size of our panel, consisting of thousands or tens of thousands of loci,
however, makes it infeasible to optimise all parameters simultaneously.
Therefore, we proceed in two steps. First, we estimate a separate base rate
for the evolution of STR length at each locus; then we aggregate into
groups of loci of the same repeat unit motif, and estimate a more involved
transition rate across the range of alleles. We then iteratively repeat
these two steps until the estimated parameters converge. We use eight
artificial lineage trees constructed from four different cell-lines to
estimate the evolutionary dynamics of STR evolution. In addition, we apply
the first step in this approach, estimating the evolution base rates, to
data simulated using the above calibrated models, in order to asses the
stability of the estimated parameters. Although there is coherence between
data from the same cell-lines, we also observe significant variations of
the model parameters among different cell-lines.

\section{STR Data}\label{sect:data}

STR data were produced by a specialized assay comprising the following
steps:
\begin{itemize}
\item[(1)] Following DNA extraction from the sample, whole genome
  amplification was performed. Different protocols were used for different
  datasets: the Repli-G protocol in the \texttt{DU145} dataset, the
  Ampli1 protocol~\cite{ampli1_2015} in the \texttt{HESC} dataset, and an
  in-house phi29-based protocol~\cite{hct_liming_unpub} in the two
  \texttt{HCT116} datasets.
\item[(2)] The amplified DNA was hybridized with a panel of duplex molecular
  inversion probes (MIP) that target genomic regions known to harbour a STR
  with unique flanking regions. For details about this panel
  see~\cite{om_lineage_2021}. In the \texttt{HESC} dataset and the two
  \texttt{HCT116} datasets, the panel OM9 from the above publication was
    used. In the \texttt{DU145} dataset, OM6 was used.
\item[(3)] The selected DNA was then gap-filled, ligated, prepared as
  standard Illumina libraries and sequenced~\cite{om_lineage_2021}.
\item[(4)] The forward and reverse reads were merged using
  \texttt{PEAR}~\cite{pear_2013}.
\item[(5)] The paired-end reads were aligned to a custom index comprising
  the same loci as the hybridization panel. For each locus, target
  sequences with different counts of repetitive units are enclosed between
  unalterered flanking sequences. Each mapped read therefore specifies a
  locus and an STR length.
\item[(6)] For each STR locus, a histogram of observed lengths was
  determined.
\end{itemize}

These STR length measurements cannot be used without corrections because
the \emph{in vitro} amplification steps during library preparation
introduce further stutter mutations. The raw STR length data therefore
follow a characteristic distribution. In previous work, a collection of
reference distributions for the same panel of loci was generated for
various lengths of the original alleles~\cite{stutter_2019}. The
distributions obtained from the sequencing data were then compared to the
reference distributions taking $1$ minus correlation as distance metric. In
order to find the best match, we used k-d tree search~\cite{Freidman:77}. A
further difficulty arose from the fact that the cells of interest are
derived of a diploid genome and hence for each locus there are (usually)
two alleles. We therefore estimated a superposition of two reference
distributions. The result of this procedure was the best estimate of the
true \textit{in vivo} STR length for each allele. In typical samples we
assayed 3000--8000 loci per sample, translating to 5000--10000 alleles.

\setlength{\tabcolsep}{6pt}
\begin{table}
  \caption{Overview of STR datasets used for model
    inference.}\label{table:tree_sizes}
  \begin{center}
    \begin{tabular}{l c c c}
      \toprule
      Tree & \# of leaves & \# of unique cells & Max.\ depth of leaf \\
      \midrule
        \texttt{DU145\_A4} & 671 & 138 & 8 \\
        \texttt{DU145\_A4\_deep} & 405 & 84 & 4 \\
        \texttt{DU145\_A4\_rest} & 266 & 54 & 7 \\
      \midrule
        \texttt{HCT116\_MSI} & 80 & 80 & 3 \\
      \midrule
        \texttt{HCT116\_MSS} & 92 & 92 & 3 \\
      \midrule
        \texttt{WIS\_A8} & 90 & 35 & 1 \\
        \texttt{WIS\_D1} & 523 & 438 & 11 \\
        \texttt{WIS\_D11} & 33 & 15 & 1 \\
      \bottomrule
    \end{tabular}
  \end{center}
\end{table}

In order to derive a model for converting STR length data into evolutionary
distances we used a collection of datasets that were specifically prepared
to investigate the use of STR length data. In brief, a culture was grown
starting from a single cell. At a predefined time point, single cells were
picked from this culture. From some of these cells DNA was extracted and STR
lengths data were measured, while other cells formed the seed for new
cultures. This process was then repeated for several generations to produce
an artificial lineage tree. See~\cite{cost_effective_lineage_2016} for more
details about the preparation. For the present study we used the following
eight datasets:
\begin{itemize}
  \item a tree origined from a single \texttt{DU145} cell
    (prostate cancer, displaying microsatellite
    instablity)~\cite{cost_effective_lineage_2016} (\texttt{DU145\_A4}) as
    well as deep sub-clade of this tree (\texttt{DU145\_A4\_deep}) and the
    complement of this sub-clade (\texttt{DU145\_A4\_rest})
  \item a tree origined from a single \texttt{HCT116} cell
    (colon cancer, displaying microsatellite
    instablity)~\cite{hct_liming_unpub} (\texttt{HCT116\_MSI})
  \item a tree origined from a single \texttt{HCT116} cell with a
    functional version of the MMR gene hMLH1
    reintroduced~\cite{hct116_mss_1994,hct_liming_unpub}
    (\texttt{HCT116\_MSS})
  \item three trees origined from three related (same ampule) single
    \texttt{HESC} cells (first published here) (\texttt{WIS\_A8},
    \texttt{WIS\_D1}, \texttt{WIS\_D11})
\end{itemize}
Sizes of the datasets are detailed in Tab.~\ref{table:tree_sizes}. For a
sketch of the structure of the trees, see Fig.~\ref{fig:tree_shapes}.

\section{Model of STR Length Evolution}

We model the temporal evolution of STR length as a continuous-time Markov
chain. A practical difficulty is that STR lengths show large variations
across loci and hence we need to consider a fairly large number of
states. For the data described in the previous section we consider $5\le
j,k\le 38$, i.e., lengths between $5$ and $38$. To avoid complicated
treatment of ``partial'' repeat units, we only count complete repeat units
and therefore take the STR length to be an integer. In order to reduce the
number of parameters that need to be estimated indepenently for each locus
we assume that we can decompose the rate matrix for a given locus $\ell$
and repeat unit motif $\tau$ into the following simple product form:
\begin{equation}
  \mathbf{R}(\ell,\tau) =  \mu(\ell) \mathbf{R}(\tau)
\end{equation}
Here $\mu(\ell)$ is a single mutation rate specific for each locus
$\ell$. In contrast, $\mathbf{R}(\tau)$ is a matrix common to all loci of
the same motif that describe the \emph{relative} rates of length
changes. Note that we allow both $\mu(\ell)$ and $\mathbf{R}(\tau)$ to
depend on the cell-line. These parameters need therefore to be
independently inferred for each dataset. We further constrain the model by
two additional plausible assumptions:
\begin{itemize}
\item[(1)] The Markov chain is reversible and hence has a stationary
  distribution;
\item[(2)] This stationary distribution coincides with the empirical
  distribution of STR lengths in the human genome.
\end{itemize}
The first assumption is clearly an approximation since STRs that become too
short will not behave like STRs any more. Since the data we actually
measure are far away from this limit, it is nevertheless a safe assumption
to make in our setting. The short time-scale of stutter mutations strongly
suggests that their length distributions are equilibrated in the human
genome. 

In order to reduce the number of parameters that need to be estimated we
use a simple model for the relative rates $\mathbf{R}_{jk}$ of a transition
from length $j$ to $k$, and assume that these values depend only on the
length of the STR and the change in length. We first construct a symmetric,
doubly-stochastic rate matrix as following:
\begin{equation}\label{eq:def_tilde_R}
  \mathbf{\tilde{R}}_{jk} :=
  \begin{cases}
    \exp(\gamma+\alpha(j+k)-\lambda|j-k|), & j \ne k \\
    - \sum_{i \ne j} \mathbf{\tilde{R}}_{ji}, & j = k
  \end{cases}
\end{equation}
where the parameter $\gamma$ is a scaling parameter. Such a symmetric rate
matrix will lead to a uniform stationary distribution. It can be adjusted
to enforce a given empirical stationary distribution by setting
\begin{equation}\label{eq:def_R}
  \mathbf{R}_{jk}(\tau) \coloneqq {s_j}^{-1} \mathbf{\tilde{R}}_{jk}(\tau)
\end{equation}
where $s_j(\tau)$ is the proportion of STRs of repeat unit motif $\tau$ of
length $j$ in the genome. A short derivation of this correction can be
found in \nameref{section:stationary_dist}. The empirical distributions for
the human genome are shown in Fig.~\ref{fig:stationary_dist}.

\begin{figure}[t]
  \begin{tabular}{cccc}
    \includesvg[width=0.225\textwidth]{figures/stationary_dist/stationary_dist_A.svg}
    &
    \includesvg[width=0.225\textwidth]{figures/stationary_dist/stationary_dist_AC.svg}
    &
    \includesvg[width=0.225\textwidth]{figures/stationary_dist/stationary_dist_AG.svg}
    &
    \includesvg[width=0.225\textwidth]{figures/stationary_dist/stationary_dist_AT.svg}
  \end{tabular}
  \caption{Stationary distributions of STR lengths for the repeat unit motifs
    \mbox{$\tau\in\{\mathtt{A},\mathtt{AC},\mathtt{AG},\mathtt{AT}\}$} in
    the reference genome \texttt{hg38}.}\label{fig:stationary_dist} 
\end{figure}

In order to efficiently fit the model to the observed STR count data for a
given sample, we proceed iteratively. We start with a uniform matrix of
transition rates $\mathbf{R}=\mathbf{J}-N\mathbf{I}$, where $\mathbf{I}$ is
the identity matrix, $\mathbf{J}$ is the matrix with entries $1$, and $N$
is the size of the matrices (in our case 34). For this transition rate
matrix, we have the following closed-form solution for the probabilities
$\mathbf{S}_{jk}$ that $j$ is substituted by $k$ after time $t$:
\begin{equation}
  \begin{split}
    \mathbf{S}_{kk}(t) &= \mathbf{S}_{=} = \frac{1}{N} \left( 1 + (N-1) \exp(-N t) \right)
    \\
    \mathbf{S}_{jk}(t) &= \mathbf{S}_{\ne} = \frac{1}{N} \left( 1 - \exp(-N t) \right)
    \quad\text{if}\ j\ne k
  \end{split}
\end{equation}
We aim to use the maximum-likelihood estimator to obtain $\mu(\ell)$ for
each locus. To this end, we compare the STRs in each pair of leaves $u$ and
$v$ in the known trees. We write $d(u,v)$ for their divergence in the tree
measured as the number of generations, more precisely, passages of the
cultures separating the two cell colonies, plus $2$ for the divergence of
each sample from the founder of its colony. For two cells within the same
colony, this number is still taken as $2$; a short derivation of this can
be found in \nameref{section:expectancy_same_colony}. Now we substitue the
time as $t = \mu(\ell) d(u,v)$ in the above solution. Moreover, we write
$n_{=}(\ell,u,v)$ and $n_{\ne}(\ell, u,v)$ for the number of equal and
distinct alleles between the two samples, respectively. Note that
$n_{=}(\ell,u,v)+n_{\ne}(\ell,u,v) = 2$ if $u$ and $v$ have two alleles
each for locus $\ell$, and $n_{=}(\ell,u,v)+n_{\ne}(\ell,u,v) = 1$
otherwise. The likelihood $\mathcal{L}(\ell)$ of observing a set pairs of
STRs at given locus $\ell$ thus takes the form
\begin{equation}\label{eq:site_log_likelihood_simple}
  \begin{array}{lll}
    \log\mathcal{L}(\ell) = \sum_{u,v} & & n_{=}(\ell,u,v)\log
      \mathbf{S}_{=}(\mu(\ell) d(u,v))
    \\
    & + & n_{\ne}(\ell,u,v) \log \mathbf{S}_{\ne}(\mu(\ell) d(u,v))
  \end{array}
\end{equation}
A binary search for a root of the derivative is sufficient to determine
$\max\log\mathcal{L}$ as a function of $\mu$.

Given the estimated values of $\mu(\ell)$, we proceed by refining the model
for the rate matrices $\mathbf{R}(\tau)$. We estimate a single matrix for
all the loci $\ell \in L(\tau)$ of each repeat unit motif $\tau$. Again
this is done for each tree separately. For a given locus $\ell$ and
two samples $u, v$, we count the number $a_{j,k}(\ell,u,v)$ of transitions
$j \to k$ from allele length $j$ in $u$ to allele length $k$ in $v$, as
follows:
\begin{itemize}
\item If both $u$ and $v$ have a single allele, we count the single
  transition $u_1 \to v_1$.
\item If both $u$ and $v$ have two distinct alleles, we order them as
  $u_1 < u_2, v_1 < v_2$, and count two transitions $u_1\to v_1$ and
    $u_2\to v_2$.
\item If $u$ has two alleles and $v$ has one, we count a single transition
  of the allele of $u$ closer to the one in $v$, i.e.\ denoting
  $i = \mathop{minarg}\{|u_1 - v_1|, |u_2 - v_1|\}$ we count $u_i \to v_1$.
  In case of a tie, we take the smaller of the two alleles in $v$.
\item If $u$ has one allele and $v$ has two, we count a single transition
  of the allele of $u$ to the one in $v$ closer to it, i.e.\ denoting
  $i = \mathop{minarg}\{|u_1 - v_1|, |u_1 - v_2|\}$ we count $u_1 \to v_i$.
  In case of a tie, we take the smaller of the two alleles in $u$.
\end{itemize}
The likelihood function for a motif then takes the form 
\begin{equation}\label{eq:r_params_log_likelihood}
  \begin{split}
    \log\mathcal{L}(\tau;\gamma,\alpha,\lambda) &= \sum_{u,v}
    \sum_{\ell \in L(\tau)} \sum_{j,k} a_{j,k}(\ell,u,v)
    \log\mathbf{S}_{jk}(\mu(\ell)d(u,v))
    \\ &\qquad\text{with } \mathbf{S}(t) =
    \exp(t\mathbf{R}(\gamma,\alpha,\lambda))
  \end{split}
\end{equation}
For each value of the three parameters $\gamma$, $\alpha$, and $\lambda$ we
can calculate the rate matrix, and from it the transition probability
matrices for the required time points, in order to evalute
$\log\mathcal{L}$. Optimization of $\log\mathcal{L}$ over the parameter
space is done using \mbox{L-BFGS-B}~\cite{l_bfgs_b_1997}; finding the
maximal log-likelihood value, we obtain the optimal parameters and our
estimation of the rate matrices $\mathbf{R}(\tau)$.

We now proceed to re-estimate the rate coefficients of $\mu(\ell)$ based on
the newly estimated rate matrices. The expression for the likelihood
\begin{equation}\label{gen_log_likelihood}
  \begin{split}
    \log\mathcal{L}(\ell) &= \sum_{u,v} \sum_{j,k} a_{j,k}(\ell,u,v)
    \log\mathbf{S}_{jk}(\mu(\ell)d(u,v))
    \\ &\qquad\text{with } \mathbf{S}(t) =
    \exp(t\mathbf{R}(\tau(\ell)))
  \end{split}
\end{equation}
is somewhat more involved than in~\eqref{eq:site_log_likelihood_simple},
but still analytical and therefore a search for the root of the derivative
can again be used. After re-estimating $\mu$, the rate matrices
$\mathbf{R}(\tau)$ are re-estimated as well
using~\eqref{eq:r_params_log_likelihood} again. We observe that a third
iteration of estimation of $\mu$ leads to no further noticable changes.

\section{Simulations}

In order to assess the error margins on the estimates which our method
provides, we simulate data using the model described above, and attempt to
reassess the parameters for these simulated data.

Given a tree topology, a transition rate matrix, a base mutation rate, and
an initial allele distribution, we generate a single simulated locus in the
following fashion: We randomly select an allele for the root node of the
tree, from the initial distribution, and then we randomly mutate that
allele along the tree topology according to the transition rate matrix and
the base mutation rate. Since our tree topologies all consist of unit
edges, the time parameter for each such progression is $1$. As our
transition matrices are reversible, no bias is introduced by starting at
the root node.

For our simulations we take the following paramters:

\begin{itemize}
\item The tree topolgies of our experimentally-generated datasets:
  \newline
  \texttt{DU145\_A4}, \texttt{DU145\_deep}, \texttt{DU145\_rest},
  \texttt{HCT116\_MSS}, \texttt{HCT116\_MSI}.
\item The transition rate matrices estimated for these trees,
  $\mathbf{R}(\tau)$, for $\tau = \mathtt{AC}, \mathtt{A}$.
\item Either:
  \begin{itemize}
  \item One of a few fixed rates, $0.3, 1, 3, 10$, which are within the
    high-density region of the rate distributions; or
  \item The entire distribution of rates $\mu(\ell)$ estimated per locus
    $\ell$ for these trees.
  \end{itemize}
\item Corresponding to the rates above, either:
  \begin{itemize}
  \item A uniform distribution of all alleles between $5$ and $38$; or
  \item The distribution of alleles in the reference genome for the panel
    of loci we are simulating.
  \end{itemize}
\end{itemize}

The first rate selection allows us to assess the error margins we can claim
for each specific rate estimate. 

The second rate selection allows us to compare the correlation values
between the rates estimated for an experimentally-generated tree and those
estimated for ones simulated using those rates, to the correlation values
between the rates estimated for different trees. This gives us an
assessment of the meaningfulness of these correlation values.
 
\section{Results}

In the top row of Fig.~\ref{fig:locus_coeffs} we can see a density plot of
the distributions of the locus-specific rates $\mu$, for the eight datasets
listed in Sect.~\ref{sect:data} (\nameref{sect:data}). They are grouped by
cell-line: the three \texttt{DU145} trees (one tree and two subtrees), the
two \texttt{HCT116} trees, and the three \texttt{HESC} trees. We note that
all curves have a similar form, suggesting some underlying common
distribution, with varying parameters. We also note some skews of the
center of this distribution, corresponding to a difference in the base rate
of mutation.

For a more detailed comparison of the rates we compute linear regressions of
shared loci for each pair of trees. We begin by taking the pearson correlation
coefficients as a measure of similarity of the distributions
(Fig.~\ref{fig:lin_reg_heatmap}, left panel). Comparing rate estimations for
trees of different cell-lines, we observe a substantial difference, reflected
through low correlations. Comparing the rate estimates for the same cell-lines,
however, we observe very high correlation for each pair \texttt{DU145} trees
and each pair of \texttt{HESC} trees, respectively. The locus-specific rates
for two \texttt{HCT116} trees exhibit a lower correlation than the
\texttt{DU145} and \texttt{HESC} data. This is consistent with the fact that
these are two related cell-lines, but with a major genetic modification
concerning mismatch repair (MMR)~\cite{hct116_mss_1994}. However the
correlations are still higher than those between completely different
cell-lines. This suggests that the common origin of the two variants of
\texttt{HCT116} still plays a significant role.

\begin{figure}[t]
  \begin{tabular}{ccc}
    \includesvg[width=0.3\textwidth]{figures/locus_coeffs_iter3/locus_coeffs_iter3_DU145_linear_raw.svg}
    &
    \includesvg[width=0.3\textwidth]{figures/locus_coeffs_iter3/locus_coeffs_iter3_HESC_linear_raw.svg}
    &
    \includesvg[width=0.3\textwidth]{figures/locus_coeffs_iter3/locus_coeffs_iter3_HCT116_linear_raw.svg}
    \\
    \includesvg[width=0.3\textwidth]{figures/locus_coeffs_iter3/locus_coeffs_iter3_DU145_linear_norm.svg}
    &
    \includesvg[width=0.3\textwidth]{figures/locus_coeffs_iter3/locus_coeffs_iter3_HESC_linear_norm.svg}
    &
    \includesvg[width=0.3\textwidth]{figures/locus_coeffs_iter3/locus_coeffs_iter3_HCT116_linear_norm.svg}
  \end{tabular}
  \caption{Distribution of locus-specific rates $\mu(\ell)$ estimated
    independently for the eight datasets, and grouped by cell-line: the
    three \texttt{DU145} trees are subsets of the same tree; the three
    \texttt{HESC} trees are separately generated; the two \texttt{HCT116}
    trees are grouped together, despite a genetic modification in the
    cell-line seeding \texttt{HCT116-MSS}. See
    also Table~\ref{table:tree_sizes}. Top row are the raw coefficients; bottom
    row they are scaled by the slope coefficients of linear regression
    within each group.}\label{fig:locus_coeffs}
\end{figure}

Using the slope of an orthogonal distance regression as an estimation for
overall scaling of rate (Fig.~\ref{fig:lin_reg_heatmap}, right panel), we
continue to the scaling of the rates estimated for trees of the same cell-line.
We note that the rates for the \texttt{DU145} trees have very little scaling,
as would be expected since the three trees are essentially subsets of the same
tree. A slightly higher $1.3$-fold scaling of the rates of
\texttt{HCT116\_MSI} as compared to the \texttt{HCT116\_MSS} tree is also
as expected, since the variants displaying microsatellite instability
should indeed be more quickly mutating. The large factor between one of the
\texttt{HESC} trees (\texttt{WIS\_D1}) and the two others
(\texttt{WIS\_A8}, \texttt{WIS\_D11}), of respectively $2.9$ and $4.7$, is
somewhat surprising, but might be an artifact arising from the shallowness
of the two latter trees, which allows proper calibration of only the
faster-evolving loci. This might also explain the lower correlation
compared to \texttt{DU145}.

Scaling the locus rate distributions by the these slopes, we obtain the
bottom row of Fig.~\ref{fig:locus_coeffs}. For \texttt{DU145} there is no
large difference, but they are more closely aligned. For \texttt{HESC} this
successfully corrects the shift seen in the top row, aligning the peaks of
the density plot, and we can see the close relation of these distributions,
despite the scaling on both axes. For \texttt{HCT116} the peaks are not
aligned, but the center of the distribution is. This highlights their lower
correlation.

\begin{figure}[t]
  \begin{tabular}{cc}
    \includesvg[width=0.45\textwidth]{figures/linreg_locus_coeffs_iter3/linreg_locus_coeffs_iter3_correlations.svg}
    &
    \includesvg[width=0.45\textwidth]{figures/linreg_locus_coeffs_iter3/linreg_locus_coeffs_iter3_slopes.svg}
  \end{tabular}
  \caption{Linear regressions of the locus-specific rate parameters
    $\mu(\ell)$ between pairs of samples: correlations (left) and estimated
    slopes (right). The slopes are column against row, so that a
    higher-than-$1$ slope means the dataset of the column mutates more
    quickly than that of the row.}\label{fig:lin_reg_heatmap}
\end{figure}

\begin{figure}[t]
  \begin{tabular}{ccc}
    \includesvg[width=0.3\textwidth]{figures/r_mat_params_iter2/r_mat_param_iter2_alpha.svg}
    &
    \includesvg[width=0.3\textwidth]{figures/r_mat_params_iter2/r_mat_param_iter2_gamma.svg}
    &
    \includesvg[width=0.3\textwidth]{figures/r_mat_params_iter2/r_mat_param_iter2_lambda.svg}
  \end{tabular}
  \caption{Optimised rate-matrix model parameters, estimated separately for
    each tree and repeat unit motif.}\label{fig:R_params}
\end{figure}

Let us now turn to the motif-specific parameters $\gamma$, $\alpha$, and
$\lambda$, as seen in Fig.~\ref{fig:R_params}. These are independently
estimated for each repeat unit motif, and we considered only data for the
four motifs $\tau \in \{ \mathtt{A}, \mathtt{AC}, \mathtt{AG},
\mathtt{AT} \}$. Although the parameters have a straightforward
interpretation in the symmetric rate model $\mathbf{\tilde{R}}$, they
are confounded by the stationary length distributions in
Fig.~\ref{fig:stationary_dist}. Despite the negative values of $\alpha$
appearing to suggest that the mutation rate decreases with repeat length, we
observe that $\mathbf{R}$ in fact contains higher values for larger lengths.

Again, we see a similar pattern of coherence among cell-lines. The
parameters for each given motif $\tau$ are very close among the
\texttt{DU145} trees, among the \texttt{HESC} trees, and also between the
two \texttt{HCT116} trees, but are very different across cell-lines.
Remarkably, the parameters for the \texttt{HCT116} trees are significantly
more similar to each other than the locus rates are. This suggests that the
common origin indeed plays a role in the length transition rates, and that
the modification in MMR gene hMLH1 might affect certain loci more than
others, but not so much the per-motif transition rates.

In order to assess the reliability of these estimates, we turn to simulated
data. First, we run multiple ($1000$) simulations with the same base rate and
observe the distributions of the rates re-estimated from these simulations.
This is done for various topolgies, transition rate matrices, and base
rates. For some simultation runs, not a single mutation is generated; these
are excluded from the plots of the distribution.

\begin{figure}[t]
  \begin{tabular}{m{0.3\textwidth} m{0.3\textwidth} m{0.3\textwidth}}
    \includesvg[width=0.3\textwidth]{figures/sim_single_rate/sim_single_rate_0_3.svg}
    &
    \includesvg[width=0.3\textwidth]{figures/sim_single_rate/sim_single_rate_1.svg}
    &
    \includesvg[width=0.3\textwidth]{figures/sim_single_rate/sim_single_rate_3.svg}
    \\
    \includesvg[width=0.3\textwidth]{figures/sim_single_rate/sim_single_rate_10.svg}
    &
    \multicolumn{2}{m{0.6\textwidth}} {
      \includesvg[width=0.6\textwidth]{figures/sim_single_rate/sim_single_rate_nonzero.svg}
    }
  \end{tabular}
  \caption{Distributions of estimated mutation site rates from simulated
  data, in natural log scale. Various transition rates and topologies were
  used, with $1000$ simulations generated for each. Horizontal lines show
  the mean of each distribution, and in black, the original rate used for
  the simulation. In the last panel, which fraction of the simulation runs
  had at least one mutation, and therefore had a non-$0$ estimated rate;
  The rest are excluded from the distributions.}\label{fig:sim_single_rate}
\end{figure}

As we can see in Fig.~\ref{fig:sim_single_rate}, all of the distributions
have their mean (plotted with a dashed colored line), as well as the peak
of the distribution, very near the base rate used for the simulation
(plotted with a solid black horizontal line), with the high-density region
of the distribution within $0.5$ to $1$ orders of magnitude away.

In the last panel we can see which fraction of the simulation runs yielded
no mutations (and therefore had their base rate estimated as $0$). These
were excluded from the other plots so as not to skew the distribution.
Again, as expected, this fraction is higher for lower base rates, smaller
trees, and slower transition matrices. The implication for the experimental
data is, that this is the fraction of loci of a given rate that we would
erroneously calibrate with a $0$ mutation rate and therefore exclude from
the analysis. However this would only lead to ignoring some of our data,
not to bias in the analysis.

\begin{figure}[t]
  \begin{tabular}{m{0.5\textwidth} m{0.4\textwidth}}
    \includesvg[width=0.5\textwidth]{figures/sim_full_panel/sim_full_panel_DU145_correlations.svg}
    &
    \includesvg[width=0.4\textwidth]{figures/sim_full_panel/sim_full_panel_HCT_correlations.svg}
    \\
    \includesvg[width=0.5\textwidth]{figures/sim_full_panel/sim_full_panel_DU145_slopes.svg}
    &
    \includesvg[width=0.4\textwidth]{figures/sim_full_panel/sim_full_panel_HCT_slopes.svg}
  \end{tabular}
  \caption{Linear regressions of the locus-specific rate parameters
  $\mu(\ell)$ between pairs, of the rates estimated from simulations and/or
  from experimental data (and used as the base for the simulations):
  grouped by cell-line, \texttt{DU145} (right) and \texttt{HCT116} (left);
  correlations (top) and estimated slopes (bottom). The slopes are column
  against row, so that a higher-than-$1$ slope means the dataset of the
  column mutates more quickly than that of the row. Marked in green are the
  comparisons of experimental datasets.}\label{fig:sim_full}
\end{figure}

Continuing to the simulations employing the entire distribution of loci, on
the top row of Fig.~\ref{fig:sim_full} we can see the correlations between
the original distributions (estimated from experimental data) and those
estimated from the simulations. Here only the blocks of similar models are
displayed; the correlation accross different models, which correspond to
different cell lines, remains low (see on the left of
Fig.~\ref{fig:sim_full_alt}). From this we can conclude, that the variance
in the estimated rates is not enough to confound our method. 

Within the blocks of similar models, we see that the correlation between
original distributions and the simulations originating from the same
distributions is comparable to that among different models belonging to the
same cell line. Therefore within the same cell line, the differences in
calibrated models might be due to variance. Notably, the correlation
between original and simulation is higher for faster transition rate
matrices. On the left of Fig.~\ref{fig:sim_full_supp} we can see another
comparison of this; the effect of the size of the tree is also noticable,
but dominated by the transition rate matrices.

Finally, we move to the slopes of the correlations among the distributions.
In the bottom row of Fig.~\ref{fig:sim_full} we can see the estimated
scaling parameters between original distributions and those reconstructed
from simulations. For the more quickly-evolving model of \texttt{DU145}, we
see almost no scaling factor between the original and the simulated
distributions. For the \texttt{HCT116} models we see a larger factor, even
larger than that between the two original \texttt{HCT116} distributions. We
note that this factor is larger for the slower of the two models
(\texttt{HCT116-MSI}). This is consistent with the expected effect of
``missing data'' of some loci which receive a $0$ mutation rate in the
simulations. Comparing to the right of Fig.~\ref{fig:sim_full_supp} we see
the same effect.

\section{Concluding Remarks}
  
The main application for estimating the rate models described in the
previous section is retrospective single-cell lineage reconstruction, in
particular of sample sets consisting of various healthy cells and cancer
cells of various stages. The variations in mutational behaviour among
cell-lines, in particular the low correlation of locus specific mutation
rates, however, suggest that cells of different differentiation states
might also vary in their mutational behaviour. This presents a problem when
attempting to reconstruct a tree composed of heterogenous samples, which
would require a unified model of mutation.

The datasets we analysed in this contribution provide us with limited
insight as to the factors governing this difference in locus specific rates
and of length transition rates. This is due to the fact that they all
belong to different individuals as well as different cell types. In order
to further investigate those factors, we will require datasets controlling
for these various factors --- trees originating from the same cell type in
different individuals, and trees originating from various cell types within
a single individual. Such datasets may allow us to determine some
commonalities and possibly develop a collection of models suitable for
reconstructing lineage of various datasets. This will require in addition
the development of a framework for estimating the likelihood of topologies
which include shifts between mutational models.

We also note the limitation of this method. While our precision in
determining the locus-specific rates (as estimated with the help of
simulations) is enough to clearly demonstrate the differences among
completely different cell-lines, between related cell-lines (the two
\texttt{HCT116} variants) the difference of the rates is close to the
resolution of our method. This suggests the necessity of datasets with a
deeper topology, which could allow for more precise inference.

In addition to the lineage reconstruction concerns, a better understanding
of the factors determining the mutational behaviour of STRs might be of
interest as of itself. Microsatellite instability (a marked increase in the
mutation rate of STRs) is a well-studied phenomenon in cancer research, and
further results of our method might contribute to this research.

\section*{Acknowledgments}
A.O. and research in the Klein lab was supported by the Deutsche
Forschungsgemeinschaft (DFG) (TRR-305/A01). Research in the Shapiro lab was
supported by the European Union grants ERC-2014-AdG (project no. 670535)
and EU-H2020-Health (project no. 874606). Research in the Stadler lab is
supported by the German Federal Ministry of Research, Technology and Space
(BMFTR) through DAAD project 57616814 (SECAI, School of Embedded Composite
AI), the German Network for Bioinformatics Infastructure, (de.NBI/RBC,
grant W-de.NBI-018), and jointly be the BMFTR and the S{\"a}chsische
Staatsministerium für Wissenschaft, Kultur und Tourismus in the programme
Center of Excellence for AI-research \emph{Center for Scalable Data
Analytics and Artificial Intelligence Dresden/Leipzig}, project
identification number: SCADS24B.

\section*{Author Contribution}
Conceptualization, A.O. and P.F.S.; formal analysis and software, A.O.; wet-lab
methodology, single-cell isolation, and sample preparation, T.M., L.T. and
T.B.; funding acquistion, E.S. and C.A.K.; wet-lab supervision, E.S.;
supervision, C.A.K and P.F.S.

\section*{Data and Code Availability}

Data for the \texttt{DU145} dataset is available as supplementary material
for~\cite{om_lineage_2021}. Additional data for that dataset, as well as
data for the \texttt{HESC} dataset, will be gladly made available upon
request by the corresponding author. Regarding the \texttt{HCT116} datasets
please contact the authors of~\cite{hct_liming_unpub} until that manuscript
is published.

Code for the various dataset perparation, parameter estimation and
simulation steps in this manuscript, will also be gladly made available
upon request by the corresponding author.

\bibliography{str_model}

\renewcommand\thefigure{S\arabic{figure}}
\setcounter{figure}{0}

\section*{Appendix A}\label{section:stationary_dist}

The following result seems to be well known although we are not aware of a
convenient reference. We therefore include a short proof for completeness.
Denote by $\mathbf{1}$ the vector with all entries $1$ and write
$\mathbf{0}$ for the zero vector. Note that any rate matrix
$\mathbf{\tilde{R}}$ satisfies $\mathbf{\tilde{R}}\mathbf{1}=\mathbf{0}$
since $\mathbf{R}_{jj} = -\sum_{k\ne j}\mathbf{R}_{jk}$ by definition.
Moreover, for a row vector $\mathbf{p}$, denote by
$\mathop{diag}{(\mathbf{p})}^{-1}$ the diagonal matrix with entries
${\mathbf{p}_j}^{-1}$. One easily checks that $\mathbf{p}
\mathop{diag}{(\mathbf{p})}^{-1}= \mathbf{1}$.

\begin{lemma}
  Let $\mathbf{\tilde{R}}$ be a symmetric rate matrix satisfying
  $\mathbf{\tilde{R}}\mathbf{1}=\mathbf{0}$, let $\mathbf{p}$ be a strictly
  positive (row) vector and set $\mathbf{R}\coloneqq
  \mathop{diag}{(\mathbf{p})}^{-1}\mathbf{\tilde{R}}$. Then $\mathbf{p}$ is a
  (left) eigenvector of $\exp(t\mathbf{R})$ with eigenvalue $\mathbf{1}$
  for all $t$.
\end{lemma}
\begin{proof}
  Symmetry of $\mathbf{\tilde{R}}$ implies
  $\mathbf{1}\mathbf{\tilde{R}}=\mathbf{0}$. The definition of $\mathbf{R}$
  yields
  $\mathbf{p}\mathbf{R} =
  \mathbf{p} \mathop{diag}{(\mathbf{p})}^{-1} \mathbf{\tilde{R}} =
  \mathbf{1} \mathbf{\tilde{R}} = \mathbf{0}$. We 
  compute
  \begin{equation*}
    \mathbf{p}\exp(t\mathbf{R}) =
    \mathbf{p}\sum_{k=0}^{\infty} \frac{t^k}{k!} \mathbf{R}^k =
    \mathbf{p}\mathbf{I} +
    \underbrace{\mathbf{p}\mathbf{R}}_{\mathbf{0}} \sum_{k=1}^{\infty}
    \frac{t^k}{k!} \mathbf{R}^{k-1} = \mathbf{p} 
  \end{equation*}
  i.e., $\mathbf{p}$ is indeed an eigenvector of $\exp(t\mathbf{R})$
  with eigenvalue $1$.
  \qed
\end{proof}
In particular, therefore, if $\mathbf{p}$ is a strictly positive
probability distribution, it is a stationary distribution of the continuous
time Markov chain with rate matrix
$\mathbf{R}=\mathop{diag}{(\mathbf{p})}^{-1}\mathbf{\tilde{R}}$, where
$\mathbf{\tilde{R}}$ is any symmetric rate matrix.

\section*{Appendix B}\label{section:expectancy_same_colony}

The expected cell division distance of two randomly selected cells in a
culture can be estimated by assuming a constant duplication rate. Recall
that each culture is started from a single cell in the experiments
described in Sect.~\ref{sect:data} (\nameref{sect:data}). Assuming all
divisions are symmetric, the culture forms in the $n$-th generation
(considering the seed cell a zeroth generation) a complete binary tree with
\mbox{$2^{n+1}-1$} cells of which $2^n$ are leaves. In this case a given
extant cell (leaf of the tree) has $2^h$ other extant cells within its
clade of height $h$, and thus $2^{h-1}$ cells of distance exactly $2h$.
Thus, denoting the depth of the entire colony as $n$, the sum of all
distances is: $2\sum_{h=1}^n h 2^{h-1} = (n-1)2^{n+1}+2$, and hence the
expected distance (divided by the number of other leaves, $2^n-1$) is
$2(n-1) + O(n2^{-n})$. Since we measure time in terms of passages of the
culture, not replications, we obtain an average distance of
$2\left(1-\frac{1}{n}\right)$ up a finite size correction of order
$2^{-n}$. For estimated $n$ values of 10--30 this is close enough to be
taken as $2$.

\clearpage
\section*{Appendix C}\label{section:supplementary_figs}

\begin{figure}[th]
  \begin{tabular}{cc}
    \includesvg[width=0.46\textwidth]{figures/trees/tree_DU145_A4_all.svg}
    &
    \includesvg[width=0.46\textwidth]{figures/trees/tree_WIS_D1.svg}
    \\
    \includesvg[width=0.225\textwidth]{figures/trees/tree_WIS_A8.svg}
    &
    \includesvg[width=0.225\textwidth]{figures/trees/tree_WIS_D11.svg}
  \end{tabular}
  \caption{Sketches of \texttt{DU145} and \texttt{HESC} trees. Nodes
    represent colonies. The number of extracted samples (including
    technical duplicates) is annotated in each node. The value in brackets
    gives the number of unique extracted cells. For \texttt{DU145}, the
    entire tree is \texttt{DU145\_A4}, the marked clade is
    \texttt{DU145\_A4\_deep}, and the tree excluding that clade is
    \texttt{DU145\_A4\_rest}. See also Tab.~\ref{table:tree_sizes} for an
    overview of the trees.}\label{fig:tree_shapes}
\end{figure}

 \begin{figure}[t]
   \begin{tabular}{cc}
     \includesvg[width=0.45\textwidth]{figures/sim_full_panel/sim_full_panel_correlations.svg}
     &
     \includesvg[width=0.45\textwidth]{figures/sim_full_panel/sim_full_panel_slopes.svg}
   \end{tabular}
   \caption{Linear regressions of the locus-specific rate parameters
   $\mu(\ell)$ between pairs, of the rates estimated from simulations and/or
   from experimental data (and used as the base for the simulations),
   correlations (left) and estimated slopes (right). The slopes are column
   against row, so that a higher-than-$1$ slope means the datasets of the
   column mutates more quickly than that of the row. Marked in green are
   the comparisons of experimental datasets.}\label{fig:sim_full_alt}
 \end{figure}

\begin{figure}[t]
  \begin{tabular}{cc}
    \includesvg[width=0.45\textwidth]{figures/sim_full_panel_supp/sim_full_panel_supp_correlations.svg}
    &
    \includesvg[width=0.45\textwidth]{figures/sim_full_panel_supp/sim_full_panel_supp_slopes.svg}
  \end{tabular}
  \caption{Linear regressions of the locus-specific rate parameters
  $\mu(\ell)$ between pairs, of the rates estimated from simulations and/or
  from experimental data (and used as the base for the simulations),
  correlations (left) and estimated slopes (right). Here we use for each
  simulation the model parameters calibrated for one tree and the topology
  of a different tree: the first name in each entry is the tree of which
  the mutation model parameters are used, the second name is the tree of
  which the topology is used. The slopes are column against row, so that a
  higher-than-$1$ slope means the datasets of the column mutates more
  quickly than that of the row. We can see a slight effect of faster
  transition rates in some models, otherwise these values match those of
  Fig.~\ref{fig:sim_full_alt}.}\label{fig:sim_full_supp}
\end{figure}

\end{document}